\documentclass[peerreview,onecolumn]{IEEEtran}
\usepackage{amsmath,amssymb,amsthm,amsfonts}
\usepackage{algpseudocode}
\usepackage{algorithm}
\usepackage{enumerate}
\usepackage{varwidth}
\usepackage{thmtools,thm-restate}
\usepackage{xcolor}
\usepackage{tikz}
\usepackage{subcaption}

\newcommand{\bS}{\mathbb{S}}

\DeclareMathOperator{\poly}{poly}
\DeclareMathOperator{\supp}{supp}
\DeclareMathOperator*{\E}{\mathbb{E}}
\DeclareMathOperator*{\median}{median}

\newcommand{\Oh}{\mathcal{O}}

\newcommand{\N}{\mathcal{N}}

\newcommand{\R}{\mathbb{R}}
\newcommand{\C}{\mathbb{C}}

\newtheorem{theorem}{Theorem}
\newtheorem{lemma}{Lemma}

\newtheorem{corollary}{Corollary}
\theoremstyle{definition}
\newtheorem{definition}{Definition}
\newtheorem{remark}{Remark}

\author{Yi~Li,~Vasileios~Nakos
\thanks{Y.\@ Li is with Nanyang Technological University, Singapore. Email: yili@ntu.edu.sg.}
\thanks{V. Nakos is with Saarland University, Germany. Email: billynak@gmail.com. Part of this work was completed while V.\@ Nakos was a Ph.D. student in Harvard University and supported in part by ONR grant N00014-15-1-2388}
}

\title{Sublinear-Time Algorithms for\\ Compressive Phase Retrieval\footnote{A preliminary version of this paper appeared in the Proceedings of \textit{International Symposium on Information (ISIT)}, 2018.}}

\begin{document}

\maketitle

\begin{abstract}

In the problem of compressed phase retrieval, the goal is to reconstruct a sparse or approximately $k$-sparse vector $x \in \C^n$ given access to $y= |\Phi x|$, where $|v|$ denotes the vector obtained from taking the absolute value of $v\in\C^n$ coordinate-wise. In this paper we present sublinear-time algorithms for a few for-each variants of the compressive phase retrieval problem which are akin to the variants considered for the classical compressive sensing problem in theoretical computer science. Our algorithms use pure combinatorial techniques and near-optimal number of measurements.
\end{abstract}

\IEEEpeerreviewmaketitle

\section{Introduction}\label{sec:intro}

In the past decade, the sparse recovery problem, or compressive sensing, has attracted considerable research interest with extensive applications and fruitful results. The problem asks to recover a signal $x\in \R^n$ (or $\C^n$) from linear measurements $y = \Phi x$ for some matrix $\Phi\in \R^{m\times n}$ (or $\C^{m\times n}$), assuming that $x$ is $k$-sparse (i.e.\@  containing only $k$ non-zero coordinates) or can be well-approximated by a $k$-sparse signal (intuitively $x$ contains $k$ large coordinates while the rest of the other coordinates are small). Often post-measurement noise is present, that is, $y = \Phi x + \nu$ for some noise vector $\nu$. The primary goal is to use as few measurements as possible. The algorithms are largely divided into two categories: geometric algorithms and combinatorial algorithms. Geometric algorithms usually use fewer measurements but run in $\poly(n)$ time while combinatorial algorithms run in sublinear time, usually $\Oh(k\poly(\log n))$ or $\Oh(k^2\poly(\log n))$, at the cost of slightly more measurements.

In recent years a closely related problem, called \emph{compressive phase retrieval}, has become an active topic, which seeks to recover a sparse signal $x\in \R^n$ (or $\C^n$) from the \emph{phaseless measurements} $y = |\Phi x|$ (or $y = |\Phi x| + \nu$ with post-measurement noise), where $|z|$ denotes a vector formed by taking the absolute value of every coordinate of $z$. The primary goal remains the same, i.e.\@ to use as fewer measurements as possible. Such types of measurements arises in various fields such as optical imaging~\cite{SECCMS15} and speech signal processing~\cite{RJ93}. There has been rich research in geometric algorithms for recovering the whole vector from phaseless measurements (see, e.g.~\cite{CSV12,CLS15a,CLS15b,GWX16,IPSV16,IVW17,goldstein2018phasemax,sun2018geometric}), as well as algorithms for the sparse phase retrieval problem that run in at least polynomial time (e.g.~\cite{jaganathan2013sparse, jaganathan2016stft, gao2016stable, sk2019structured}), mostly based on semidefinite programming. On the other side, there have been relatively few sublinear time algorithms --  \cite{CBJC14,IVW16,PYLR17,Nakos} are the only algorithms to the best of our knowledge. This comes in contrast to the standard sparse recovery problem with linear measurements, which sees a long history of sublinear-time algorithms (e.g.~\cite{ipw11,ps12,hikp12a, gnprs13,ikp14,k16,glps17,k17,LNW17,nswz18,ns19}).
For the case of phaseless measurements, most existing algorithms consider sparse signals, and thus such sublinear time algorithms have a flavour of code design, akin to Prony's method. Among the sublinear-time algorithms, \cite{CBJC14} considers sparse signals only, \cite{PYLR17} considers sparse signals with random post-measurement noise, \cite{IVW16} allows adversarial post-measurement noise but has poor recovery guarantee, \cite{Nakos} considers near-sparse real signals (signals of real coordinates) with no post-measurement noise but achieves constant-factor approximation and thus outperforms all other sublinear-time algorithms for real signals. The approach in \cite{Nakos} employs combinatorial techniques more widely used in the theoretical computer science literature for the classical sparse recovery problem. In this paper, we aim to improve on~\cite{Nakos} for complex near-sparse signals using similar combinatorial techniques.

More quantitatively, suppose that the decoding algorithm $\mathcal{R}$, given input $y=|\Phi x|$, outputs an approximation $\hat x$ to $x$, with the guarantee that the approximation error $d(x,\hat x)$ is bounded from above. For $x\in \C^n$, the approximation error $d(x,\hat x) = \min_{\theta\in[0,2\pi)} \|x-e^{i\theta}\hat x\|$. Specifically we consider the following two types of error guarantee:
\begin{itemize}
	\item ($\ell_\infty/\ell_2$) $\min_{\theta\in[0,2\pi)} \|x-e^{i\theta}\hat x\|_\infty \leq \frac{1}{\sqrt{k}} \|x_{-k}\|_2$;
	\item ($\ell_2/\ell_2$) $\min_{\theta\in[0,2\pi)} \|x-e^{i\theta}\hat x\|_2  \leq (1+\epsilon) \|x_{-k}\|_2$,
\end{itemize}
where $x_{-k}$ denotes the vector formed by zeroing out the largest $k$ coordinates (in magnitude) of $x$. Note that when $x$ is noiseless, that is, when $x_{-k} = 0$, both guarantees mean exact recovery of $x$, i.e., $\hat x = x$.

Besides the error guarantees, the notions of for-all and for-each in the sparse recovery problems also extend to the compressive phase retrieval problem. In a for-all problem, the measurement matrix $\Phi$ is chosen in advance and will work for all input signals $x$, while in a for-each problem, the measurement matrix $\Phi$ is usually random such that for each input $x$, a random choice of $\Phi $ works with a good probability.

In the next subsection we shall give an overview of our sublinear-time results. Our results are all for-each results. As the problem has received little attention in the community of theoretical computer science, our results are the first and preliminary along the combinatorial approach though require mild assumptions on the possible phases of the coordinates of $x$. Those assumptions are automatically satisfied when the possible phases, as commonly seen in applications (see, e.g., \cite{PYLR17}), belong to a set $P\subseteq \bS^1$ which is equidistant with a constant gap, that is, up to a rotation, $P \subseteq \{e^{2\pi i\frac{j}{m}}\}_{j=0,\dots,m-1}$ for some constant $m$. We leave the case of general complex signals as an open problem. A major difficulty is that for heavy hitters of large magnitude, a small error in the phase estimate could incur a lot of error in the overall approximation.

\subsection{Results}
In this section we give an overview of the sublinear-time results which we have obtained for the sparse recovery problem with phaseless measurements. 

First, we consider the case of noiseless signals. Similar to the classical sparse recovery where $\Oh(k)$ measurements suffice for noiseless signals by Prony's method~\cite{prony}, it is known that $\Oh(k)$ phaseless measurements also suffices for exact recovery (up to rotation) and the decoding algorithm runs in time $\Oh(k\log k)$~\cite{CBJC14}. Their algorithm is based on a multi-phase traversal of a bipartite random graph in a way such that all magnitudes and all phases are recovered by resolving multi-tons. We prove a result with the same guarantee, but our algorithm takes a different route using more basic tools and being less technically demanding. Apart from being significantly simpler, it also can be modified so that it trades the decoding time with the failure probability; see Remark~\ref{rem:noiseless_tradeoff}. 

\begin{restatable}[noiseless signals]{theorem}{noiselesstheorem}\label{thm:noiseless} 
There exists a randomized construction of $\Phi \in \mathbb{C}^{m \times n}$ and  a deterministic decoding procedure $\mathcal{R}$ such that for any signal $x\in \C^n$ with $|\supp(x)|\leq k$, the recovered signal $\hat x = \mathcal{R}(\Phi, |\Phi x|)$ satisfies that $\hat x = e^{i\theta}x$ for some $\theta \in [0,2\pi)$ with probability at least $1 - 1/\poly(k)$, where $\Phi$ has $m = \Oh(k)$ measurements and $\mathcal{R}$ runs in time $\Oh(k \log k)$.
\end{restatable}

The next results refer to approximately sparse signals and improve upon the previous ones with various degrees. For the $\ell_\infty/\ell_2$ problem our result, stated below, improves upon \cite{Nakos} in terms of the error guarantee and the decoding time. It requires a modest assumption on the pattern of the valid phases of the heavy hitters as defined below, which is often satisfied in applications where the valid phases lie in a set of equidistant points on $\bS^1$. Throughout this paper we identify $\bS^1$ with $[0,2\pi)$ or $[-\pi,\pi)$ and assume both the non-oriented distance $d(\cdot,\cdot)$ and the oriented distance $\vec{d}(\cdot,\cdot)$ on $\bS^1$ are circular. We shall also use $[m]$ to denote the set $\{1,\dots,m\}$ for any positive integer $m$, a conventional notation in computer science literature.

\begin{definition}[$\eta$-distinctness] Let $P=\{p_1,\dots,p_m\}$ be a finite set on $\bS^1$. We say $P$ is $\eta$-distinct if the following conditions hold:
\begin{enumerate}[(i)]
\item $d(p_i,p_j)\geq \eta$ for all distinct $i,j\in [m]$;

\item
 it holds for every pair of distinct $i,j\in [m]$ that
	\[
		\max_{\ell\in [m]} d(p_\ell + p_j - p_i, P) \in \{0\}\cup[\eta,\pi].
	\]
\end{enumerate}
\end{definition}
Intuitively, (i) means that the phases are at least $\eta$ apart from each other, and (ii) means that if we rotate the set $P$ of the valid phases to another set $P'$ such that some valid phase coincides with another one (in the expression above $p_i$ is rotated to the position of $p_j$), then either $P=P'$ or there exists an additive gap of at least $\eta$ around some phase. This precludes the case where $P$ is approximately, but not exactly, equidistant.

\begin{definition}[head]
Let $x\in \C^n$. Define $H_k(x)$ to be (a fixed choice of) the index set of the $k$ largest coordinates of $x$ in magnitude, breaking ties arbitrarily.
\end{definition}

\begin{definition}[$\epsilon$-heavy hitters]
Let $x\in \C^n$. We say $x_i$ is an $\epsilon$-heavy hitter if $|x_i|^2\geq \epsilon \|x_{-1/\epsilon}\|_2^2$.
\end{definition}

\begin{definition}[phase-compliant signals] 
Let $x\in \C^n$. Let $P\subseteq \bS^1$ be a set of possible phases and $T$ be the set of all $(1/k)$-heavy hitters in $T$. We say that $x$ is $(k,P)$-compliant if $\{\arg x_i: i\in T\}\subseteq P$.
\end{definition}


\begin{theorem}[$\ell_{\infty}/\ell_2$ with optimal measurements] \label{thm:ell_infty_ell_2}
There exists a randomized construction of $\Phi \in \mathbb{C}^{m \times n}$ and a deterministic decoding procedure $\mathcal{R}$ such that for $x\in \C^n$ which is $(\Oh(k),P)$-compliant for some $\eta$-distinct $P\subset \bS^1$, the recovered signal $\hat x = \mathcal{R}(\Phi, |\Phi x|, P)$ satisfies the $\ell_\infty/\ell_2$ error guarantee with probability at least $1-1/\poly(n)$, and $\Phi$ has $m = \Oh((k/\eta)\log n)$ rows and $\mathcal{R}$ runs in time $\Oh(k/\eta + k \poly(\log n))$.
\end{theorem}

It is clear that the lower bound for the traditional compressive sensing problem is also a lower bound for the compressive phase retrieval problem, and it is known that the $\ell_\infty/\ell_2$ compressive sensing problem requires $\Omega(k\log n)$ measurements~\cite{bipw10}. Therefore, the theorem above achieves the optimal measurements up to a constant factor when $\eta$ is a constant.

An immediate corollary of the $\ell_\infty/\ell_2$ sparse recovery algorithm is an $\ell_2/\ell_2$ sparse recovery algorithm, stated below, which improves upon \cite{Nakos} in approximation ratio (from a constant factor to $1+\epsilon$). 

\begin{corollary}[$\ell_2/\ell_2$ with near-optimal measurements]
There exists a randomized construction of $\Phi \in \mathbb{R}^{m \times n}$ and a deterministic decoding procedure $\mathcal{R}$ such that for $x\in \C^n$ which is $(\Oh(k),P)$-compliant for some $\eta$-distinct $P\subset \bS^1$, the recovered signal $\hat x = \mathcal{R}(\Phi, |\Phi x|, P)$  satisfies the $\ell_2/\ell_2$ error guarantee with probability at least $1-1/\poly(n)$, and $\Phi$ has $m = \Oh((k/\min\{\eta,\epsilon\}) \log n)$ rows and $\mathcal{R}$ runs in time $\Oh((k/\min\{\eta,\epsilon\}) \poly(\log n))$.
\end{corollary}

It is also known that the classical compressive sensing problem with for-each $\ell_2/\ell_2$ error guarantee and constant failure probability requires $\Omega((k/\epsilon)\log(n/k))$ measurements~\cite{PW11}. Our result above achieves the optimal number of measurements up to a logarithmic factor.

%
%

\section{Toolkit}


%

The first two results concern heavy hitters, one for estimating the value of a heavy hitter and the other for finding the positions of the heavy hitters.

\begin{theorem}[\textsc{Count-Sketch},~{\cite{CCFC}}]\label{thm:CS}
There exist a randomized construction of a matrix $\Phi \in \mathbb{R}^{m \times n}$ with $m = \Oh(K \log n)$ and a deterministic algorithm $\mathcal{R}$ such that given $y = |\Phi x|$ for $x\in \C^n$, with probability at least $1-1/\poly(n)$, for every $i\in [n]$, the algorithm $\mathcal{R}$ returns in time $\Oh(\log n)$ an estimate $|\hat x_i|$ such that
\[
||x_i|-|\hat x_i||^2\leq \frac{1}{K}\|x_{-K}\|_2^2.
\]
\end{theorem}

\begin{theorem}[Heavy hitters,~{\cite{larsen2016heavy}}]\label{thm:HH}

There exist a randomized construction of a matrix $\Phi \in \mathbb{R}^{m \times n}$ with $m = \Oh(K \log n)$ and a deterministic algorithm $\mathcal{R}$ such that given $y = |\Phi x|$ for $x\in \C^n$, with probability at least $1-1/\poly(n)$ the algorithm $\mathcal{R}$ returns in time $\Oh(K \cdot \poly(\log n))$ a set $S$ of size $\Oh(K)$ containing all $(1/K)$-heavy hitters of $x$.
\end{theorem}
We remark that the paper \cite{larsen2016heavy} does not consider complex signals but the extension to complex signals is straightforward. The algorithm is not designed for the phaseless sparse recovery either, the identification algorithm nevertheless works when the measurements are phaseless because it only relies on the magnitudes of the bucket measurements; see Theorem $2$ and Section B in \cite{larsen2016heavy}. Estimating the values of the candidate coordinates requires knowing the phases of the measurements but our theorem above does not concern this part. 

\begin{theorem}[{\cite{erdos_renyi}}]\label{thm:connected}
Let $V$ be a set of $n$ vertices. There exists an absolute constant $\kappa$ such that $\kappa n \log n$ uniform samples of pairs of distinct vertices in $V$ induce a connected graph with probability at least $1 - 1/\poly(n)$.
\end{theorem}

The following lemmata will be crucial in the analysis of our algorithms. 

%



\begin{lemma}\label{lem:phase_test}
Let $\theta_0 \in (0, \pi/12)$ be a constant. There exist constants $\epsilon_0, c_0, c > 0$ such that the following holds.
Suppose that $x,y,n_1,n_2,n_3\in \C$ such that 
$|n_1|,|n_2|,|n_3|\leq \epsilon\min\{|x|,|y|\}$ for some $\epsilon\leq \epsilon_0$. Denote by $\theta$ be the phase difference between $x$ and $y$. Then given the norms
\[
|x+n_1|, |y+n_2|, |x+y+n_1+n_2+n_3|,
\]
we can recover $|\theta|$ up to an additive error of $c_0\sqrt{\epsilon}$. Furthermore, we can recover $|\theta|$ up to an additive error of $c\epsilon$ whenever $|\theta|\in (\theta_0,\pi-\theta_0)$ .
\end{lemma}
\begin{proof}
If we know $|x|$, $|y|$ and $|x+y|$, it follows from the Law of Cosines that
\[
\cos(\pi-|\theta|) = \frac{|x|^2 + |y|^2 - |x+y|^2}{2|x|\cdot|y|} = \frac{-\Re x\bar y}{|x|\cdot|y|}.
\]
Let $x' = x + n_1$ and $y' = y+n_2$ then $x + y + n_1 + n_2 + n_3 = x' + y' + n_3$. Suppose the phase difference between $x'$ and $y'$ is $\theta'$, then we would pretend $x'+y'+n_3$ to be $x'+y'$. We compute
\[
\xi = \frac{|x'|^2 + |y'|^2 - |x'+y'+n_3|^2}{2|x'|\cdot|y'|} = \frac{-\Re x'\overline{y} - \Re x'\overline{n_3} - \Re y\overline{n_3} - |n_3|^2}{|x'|\cdot|y'|}.
\]
Hence
\[
|\xi -\cos(\pi-|\theta'|)| \leq \frac{|x'||n_3|+|y'||n_3|+|n_3|^2}{|x'|\cdot |y'|}\leq \frac{\epsilon}{1-\epsilon} + \frac{\epsilon}{1-\epsilon} + \frac{\epsilon^2}{1-\epsilon^2}\leq 3\epsilon,
\]
provided that $\epsilon \leq 1/4$.

Similarly we have
\[
\cos(\pi-|\theta'|) = \cos(\pi-|\theta|) \cdot \frac{|x|}{|x+n_1|}\cdot \frac{|y|}{|y+n_2|} + \nu,\quad |\nu|\leq 3\epsilon,
\]
and thus
\[
\cos(\pi-|\theta'|)-\cos(\pi-|\theta|) = \cos\theta\left(\frac{|x|}{|x+n_1|}\cdot\frac{|y|}{|y+n_2|}-1\right) + \nu.
\]
Note that $\frac{|x|}{|x+n_1|}, \frac{|y|}{|y+n_2|}\in [\frac{1}{1+\epsilon}, \frac{1}{1-\epsilon}]$,  it follows that
\[
|\cos(\pi-|\theta'|)-\cos(\pi-|\theta|)| \leq c_1\epsilon
\]
for some absolute constant $c_1 > 0$, and thus
\[
|\xi - \cos(\pi-|\theta|)| \leq (c_1+3)\epsilon.
\]
Therefore we can estimate $|\theta|$ up to an additive error of $c_0\sqrt{\epsilon}$ for some absolute constant $c_0 > 0$; and furthermore, there exists an absolute constant $c > 0$ such that when $\theta_0 = \pi/9$, $\epsilon$ is small enough and $|\theta| \in (\theta_0, \pi-\theta_0)$, it holds that
\[
|\arccos\xi - |\theta|| \leq c\epsilon.\qedhere
\]

\end{proof}

\begin{lemma}\label{lem:phase_difference}
Let $x,y,n_1,n_3,\epsilon,\theta_0,c$ be as in Lemma~\ref{lem:phase_test} and further assume that $\epsilon < \theta_0/(2c)$.  Suppose that $\arg y = \arg x + \theta$ for some $\theta\in [0,2\pi)$, where addition is modulo $2\pi$. Given the norms
\[
|x+n_1|, |y|, |x+y+n_1+n_3|, \left|x+\beta y+n_1+n_3\right|, \quad \beta = e^{i\theta_0},
\]
we can recover $\theta$ up to an additive error of $c\epsilon$, provided that $\theta \in (2\theta_0,\pi-2\theta_0)\cup(\pi+2\theta_0,2\pi-2\theta_0)$.
\end{lemma}
\begin{proof}
By Lemma~\ref{lem:phase_test}, we can recover $|\theta|$ up to an additive error of $c\epsilon$ when $|\theta|\in (\theta_0,\pi-\theta_0)$. To determine the sign, we rotate $y$ by angle $\theta_0$ and test the angle between $x$ and this rotated $y$ again by Lemma~\ref{lem:phase_test}. Suppose that the angle between $x$ and $\beta y$ is $\phi$ and we have an estimate of $|\phi|$ up to an additive error of $c\epsilon$, provided that $|\phi|\in (\theta_0,\pi-\theta_0)$, which is satisfied when $|\theta|\in (2\theta_0,\pi-2\theta_0)$. It holds in this case that 
\begin{gather*}
|\phi| - |\theta| = \begin{cases}
						\theta_0,& \theta > 0;\\
						-\theta_0,& \theta < 0.	
					\end{cases}
\end{gather*}
The left-hand side is approximated up to an additive error of $2c\epsilon$ and thus we can distinguish the two cases.
\end{proof}

\begin{lemma}[relative phase estimate]\label{lem:all_phase_difference}
Let $x,y,n_1,n_3,\epsilon,c_0,\theta_0,c$ be as in Lemma~\ref{lem:phase_difference} and further assume that $\epsilon < \min\{\theta_0/(2c),(\theta_0/(2c_0))^2\}$. Suppose that $\arg y = \arg x + \theta$ for some $\theta\in [0,2\pi)$, where addition is modulo $2\pi$. Given the norms
\[
|x+n_1|, |y|, |x+e^{i(\theta_0 j + \frac{\pi}{2}\ell)}y+n_1+n_3|, \quad j, \ell = 0,1
\]
we can recover $\theta$ up to an additive error of $c\epsilon$.
\end{lemma}
\begin{proof}
From Lemma~\ref{lem:phase_difference}, we know that we can recover $\theta$ up to an additive error of $c\epsilon$ when $\theta\in I$, where $I = (2\theta_0,\pi-2\theta_0)\cup(\pi+2\theta_0,2\pi-2\theta_0)$. We accept the estimate if the estimate is in the range of $I' := (3\theta_0, \pi-3\theta_0)\cup (\pi+3\theta_0,2\pi-3\theta_0)$. It follows from the assumption of $\epsilon$ that if we ever accept the estimate, $|\theta|$ must be contained in $(\frac{5}{2}\theta_0,\pi-\frac{5}{2}\theta_0)$ and we have therefore an $c\epsilon$ additive error.

Consider the phase difference between $x$ and $e^{i\pi/2}y$ and suppose that $\arg(e^{i\pi/2}y) = \arg x + \phi$, then we can recover $\phi$ up to an additive error of $c\epsilon$ for $\phi \in I$, that is, for $\theta\in J : = (\pi/2+2\theta_0,3\pi/2-2\theta_0)\cup(-\pi/2+2\theta_0,\pi/2-2\theta_0)$, which is $I$ rotated by $\pi/2$. We accept the estimate when it is in the range of $J' := (\pi/2+3\theta_0,3\pi/2-3\theta_0)\cup(-\pi/2+3\theta_0,\pi/2-3\theta_0)$.

Note that $I'\cup J'$ covers the whole $\bS^1$ when $\theta_0 < \pi/12$.
\end{proof}

\section{Noiseless Signals}

We shall need the following theorem from~\cite{Nakos}, which shows that one can recover an exactly $K$-sparse signal up to a global phase using $\Oh(K)$ measurements and in time $\Oh(K^2)$.\footnote{The runtime is dominated by Prony's method, which can be implemented in $O(K^2)$ time using a Lanczos process (see, e.g.~\cite{vandermonde_factorization}) and solving a linear Vandermonde system (see, e.g.~\cite{vandermonde_system}). A detailed discussion can be found in the arXiv version of~\cite{Nakos}.}
%

\begin{theorem}[{\cite{Nakos}}]\label{thm:noiseless_nakos}
There exists a matrix $M \in \mathbb{C}^{(6k-2) \times n}$, such that given given $y = | M x|$ for $x\in \C^n$ such that $\|x\|_0\leq K$ we can recover $x$ up to a rotation in time $\Oh(K^2)$.
\end{theorem}

We are now ready to prove Theorem~\ref{thm:noiseless}, which we restate below.

\noiselesstheorem*

Let us first give some intuition on the algorithm and the proof. We hash all $n$ coordinates to $\Theta(k/\log k)$ buckets, and in each bucket we apply the algorithm in Theorem~\ref{thm:noiseless_nakos} with $K = O(\log k)$. By standard concetration bounds there exist at most $O(\log k)$ coordinates of $\supp(x)$ in each bucket, and hence, conditioned on this event, we can find them deterministically along with their relative phases. Next we find all the relative phases by finding the relative phases ``between buckets'', which means finding the relative phases between representatives from each bucket. To that end, we build a random graph with $O(k/\log k)$ buckets and $O(k)$ edges, where each vertex represents a bucket and each edge has a label that indicates the relative phase between its two endpoints. If the graph is connected, a standard traversal will give the desired result. We show how to sample random edges from this graph by performing a careful subsampling on the vector $x$. The number of measurements and the running time follows by standard results in random graph theory. Below is the formal proof.

 \begin{proof}
Let $B = k/(c\log k)$ and $h: [n] \rightarrow [B]$ a random hash function, where $c$ is a constant. We hash all $n$ coordinates into $B$ buckets using $h$. It is a typical application of Chernoff bound that the buckets have small size, more specifically, for some constant $\bar c > c$,
\begin{equation}\label{eqn:small_bucket}
\Pr\left\{ \exists j \in [B]: |h^{-1}(j) \cap \mathrm{supp}(x)| > \bar c \log k \right\} \leq \frac{1}{\poly(k)}.	
\end{equation}
To see this, for $i \in [n]$ and $j \in [B]$, let $X_{i,j}$ be the indicator variable of the event $h(i) = j$. Then $\mathbb{E}X_{i,j} = 1/B$, and thus $\E\sum_{i\in\supp(x)} X_{i,j} = k/B = \bar c\log k$. Note that $X_{i,j}$ are negatively associated, thus the Chernoff bound can be applied, which, with appropriate constants, yields that for each $j\in [B]$,
\[
\Pr\left\{ |h^{-1}(j)\cap \supp(x)| > 5\log k \right\} = \Pr\left\{ \sum_{i\in \supp(x)} X_{i,j} > \bar c\log k\right\} \leq \frac{1}{\poly(k)}.
\]
Take a union bound over all $j \in [B]$ gives \eqref{eqn:small_bucket}.

In each bucket we run the algorithm of Theorem~\ref{thm:noiseless_nakos} with $K = \bar c\log k$. The number of measurements used for each bucket is $\Theta(\log k)$. The failure probability is $2^{-\Theta(\log k)}$, and this allows us to take a union-bound to conclude correct execution over all $B$ buckets. For each $j \in [B]$, we can find $x_{h^{-1}(j)}$ up to a global phase. We set $\supp(x) = \bigcup_{j=1}^B \supp(x_{h^{-1}(j)})$ and, for notational convenience, let $s = |\supp(x)|$ and $\ell$ be such that $2^{\ell-1} < s \leq 2^\ell$.

Now, let $B' = 2^\ell/(c \ell)$ and let $h':[n]\to [B']$ be a random hash function. We hash all $n$ coordinates into $B'$ coordinates. Similar to the above, with probability $\geq 1 - 2^{-\Omega(\ell)}$ every bucket contains at most $\bar c\ell$ nonzero coordinates of $x$. In each bucket, we run the algorithm of Theorem~\ref{thm:noiseless_nakos} with $K=\bar c\ell$, so that in $\Oh(\ell)$ measurements we can find $x_{(h')^{-1}(j)}$ up to a global phase for each $j\in [B']$, so it remains to find the relative phases across different $x_{(h')^{-1}(j)}$. 

Let $F_\ell$ be a matrix of $n$ columns and $2\alpha\kappa 2^\ell$ rows, where $\alpha$ is a sufficiently large constants and $\kappa$ is the same constant in Theorem~\ref{thm:connected}. The rows are split into $\alpha \kappa2^\ell$ groups, each has $2$ rows. For each group $j$, the first row $(F_\ell)_{2j-1}$ has independent random $\{0,1\}$-entries such that $\E[ (F_\ell)_{2j-1,t}] = 2^{-\ell}$. The second row has the form $(F_\ell)_{2j,t} = \rho_{j,t}\cdot (F_\ell)_{2j-1,t}$, where $\{\rho_{j,t}\}$ are independent uniform random variables in $\{1,i\}$.

Next we describe how to recover the relative phases across different $x_{(h')^{-1}(j)}$. Define a row index set $\hat J$ to be
\[
\hat J = \left\{j: \left|\supp((F_\ell)_{5(j-1)+1})\cap \supp(x)\right| = 2\right\}.
\]
Observe that for each $j$,
\[
\Pr\{j\in \hat J\} = 
\binom{s}{2} \frac{1}{2^{2\ell}} \left(1-\frac{1}{2^\ell}\right)^{s-2} \geq \frac{1}{2}\left(\frac{s-1}{2^\ell}\right)^2 \left(1-\frac{1}{s}\right)^{s-2}\geq \frac{1}{2}\left(\frac 12\right)^2\frac{1}{e},
\]
that is, each $j$ is contained in $\hat J$ with probability at least an absolute constant. For each such $j$, let $\{u_j, v_j\} = \supp((F_\ell)_{2j-1})\cap \supp(x)$. It is clear that $\rho_{j,u_j} \neq \rho_{j,v_j}$ with probability $1/2$. When this event happens, we say $j$ is good. We have shown that each $j$ is good with probability at least an absolute constant. Let $J\subseteq \hat J$ denote the set of good $j$'s.

We shall focus on the measurements corresponding to the groups in $J$. From each $j\in J$ we can obtain a random pair $\{u_j,v_j\} \subseteq \supp(x)$ and, moreover, $(h'(u_j),h'(v_j))$ is uniformly random on $[B']\times [B'] \setminus \{ (b,b), b\in [B'] \}$ (in other words, such a measurement corresponds to a random pair of buckets). Since $j$ is good, we obtain both $|x_{u_j} + x_{v_j}|$ and $|x_{u_j} + i x_{v_j}|$. Since we also know $|x_{u_j}|$ and $|x_{v_j}|$, we can infer the relative phase between $x_u,x_v$. The relative phases we obtain are always correct since the signal is noiseless. Let $\mathcal{M}$ be the ordered set of such pairs $(u_j,v_j)$ along with the label that we obtain about the relative phase between $u_j$ and $v_j$. By taking $\alpha$ to be big enough, we have $|\mathcal{M}| = |J|\geq \kappa B'\log B'$ (since $B'\log B' = \Theta(2^\ell)$) with probability at least $1-2^{-\Omega(|J|)}$. We run a depth-first search on $\mathcal{M}$ to infer the relative phases. We infer all the relative phases correctly when the $\mathcal{M}$ is connected, which, by Theorem~\ref{thm:connected}, happens with
probability at least $1-1/\poly(B') \geq 1 - 2^{-\Omega(\ell)}$.

So far we have shown that with probability at least $1-2^{-\Omega(\ell)}$ we can recover $\hat x$ up to a global phase difference. To improve the success probability to $1-1/\poly(k)$, we can repeat the procedure above by $\Theta(1)$ times for $\ell > \lceil \frac{1}{2}\log k\rceil$ and $\Theta(\log k)$ times for $\ell \leq \lceil \frac{1}{2}\log k\rceil$ and take the relative phase pattern that appears most often.

\paragraph*{Number of Measurements} The total number of measurements is therefore (taking a union over all possible values of $\ell=1,2,\dots,\lceil\log k\rceil$)
\[ 
\frac{k}{c\log k}\cdot \Theta(\log k) + \sum_{\ell > \lceil \frac{1}{2}\log k \rceil} \left(\frac{2^\ell}{c\ell} \cdot \Theta(\ell) + \alpha \kappa 2^\ell\right)\cdot\Theta(1) + \sum_{\ell\leq \lceil \frac12 \log k \rceil} \left(\frac{2^\ell}{c\ell} \cdot \Theta(\ell) + \alpha \kappa 2^\ell\right)\cdot\Theta(\log k) = \Oh( k )
\]
as desired. 

\paragraph*{Runtime} Recovering the in-bucket signals takes time $\Oh(B \log^2 k (\cdot \poly(\log \log k)) = \Oh(k\log k \cdot \poly(\log \log k))$. Finding the set $J$ takes time $\Oh(k)$ (see Remark~\ref{rem:runtime}). Thus, it takes time $\Oh(k + B'\log B') = \Oh(k)$ to build a pattern of relative phase differences and $\Oh(k\log k)$ time to build all patterns\footnote{In fact, building the patterns can be done in $O(k)$ time, since the $\log k$ factor due to the number of repetitions occurs only for ``small'' $\ell$, that is, $\ell \leq \lceil \frac{1}{2}\log k \rceil$. Nevertheless, this will not make any difference in our final running time.}. Finding the most frequent pattern can be implemented with a sorting followed by a linear scan. There are $\Theta(\log k)$ patterns and the comparison of two patterns takes time $\Oh(B)$. Hence finding the most frequent pattern takes time $\Oh(\log k \log\log k \cdot B + B \log k) = \Oh(k\log\log k)$.  The overall runtime is therefore $\Oh(k\log k)$.
\end{proof}

\begin{remark}\label{rem:noiseless_tradeoff}
Note that if we hash to $k^{1-\alpha}$ buckets, solve in each bucket and then combine the buckets, we can obtain a failure probability at most $\exp(- k^\alpha)$ and a running time of $\Oh(k^{1+\alpha})$. This is a trade-off between decoding time and failure probability that the previous algorithms did not achieve.
\end{remark}

\begin{remark}\label{rem:runtime}
We show how to implement efficiently the routine which finds the set of rows of $F$ whose support intersects $\supp(x)$ at  exactly two coordinates. For $\ell > \lceil \frac{1}{2} \log k \rceil $, in each repetition, the expected number of rows of $F_{\ell}$ containing an index $i \in \supp(x)$ is $ 2^{-\ell} \alpha \kappa 2^{\ell} = \alpha \kappa$. So the probability that there are more than $2\alpha \kappa 2^{\ell}$ pairs $(i,q)$ such that $i \in \supp(x) \cap \supp(F_{\ell})_q$ is $\exp(-\Omega(2^{\ell})) < 1/\poly(k)$. A similar result can be obtained for $ \ell \leq \lceil\frac{1}{2} \log k\rceil$. Suppose that $F_\ell$ is stored using $n$ lists of nonzero coordinates in each column, we can afford to iterate over all such pairs $(i,q)$, keep an array $C[q]$ that holds the cardinality of $\supp(x)\cap \supp((F_{\ell})_q)$. At the end, we find the values of $q$ with $C[q] = 2$. This implementation makes the algorithm run in $\Oh( k \log k)$ time.
\end{remark}

\section{ $\ell_{\infty}/\ell_2$ Algorithm}\label{sec:ell_inf/ell_2}

We first give a high-level overview of our algorithm. Invoking Theorems~\ref{thm:CS} and~\ref{thm:HH}, we can find a set $S$ of size $O(k)$ that contains all heavy hitters. It remains to find their relative phases. Our approach is very different from the previous section; we shall exploit the presence of the rest $n- O(k)$ coordinates to find the relative phases among the heavy hitters. We downsample the signal at the rate of $\Theta(1/k)$, and combine with gaussians to form $L = \sum_{i \in T} g_i x_i$, where $T$ is the support of the downsampled signal. With constant probability, $T \cap S = \emptyset$ (note that we do not know $S$ when designing the measurements, we only know $T$). We then hash every coordinate in $[n]\setminus T$ to $\Theta(k)$ buckets, combine with random rotations and add $L$ to each bucket. We repeat with fresh randomness $\Theta( \log n)$ times. Then, for every $i \in S$ we can find the buckets in which it is isolated from $S\setminus \{ i\}$ and its relative phase with respect to $L$ (note that $L$ is present in every bucket). This allows us to find the relative phase among all $i \in S$. The risk is that $L$ might be $0$, in which case the argument above would not work. In this case, however, $T$ would be an empty set, which in turn implies that $x$ would be $O(k)$-sparse, and hence we could run the algorithm from the previous section.

\bigskip

The remaining of the section is devoted to the details of our algorithm and its analysis.

We set $\epsilon_0,c_0,c$ to be the constants in Lemma~\ref{lem:all_phase_difference} and $\epsilon = \min\{\epsilon_0,\eta/(5c),\pi/(25c), \pi^2/(145c_0^2)\}$. Let $P\subseteq \bS^1$ be $\eta$-distinct and suppose that it contains the phases of all $1/(\tilde Ck)$-heavy hitters for some (large) constant $\tilde C$.

We first describe our construction of the measurement matrix $\Phi$ and then present the analysis and the recovery algorithm. Let $R = c_R\log n$ for some constant $c_R$ to be determined. The overall sensing matrix $\Phi$ is a layered one as
\[
\Phi = \left[ 
		\begin{array}{c}
			\Phi_{\textsc{HH}} \\
			\Phi_{\textsc{CS}} \\
			\Phi_{\textsc{L}} \\
			\Phi_1\\
			\vdots \\
			\Phi_R
		\end{array}
	\right].
\]
Here
\begin{itemize}
	\item $\Phi_{\textsc{HH}}$ is the sensing matrix in Theorem~\ref{thm:HH} with $K=k$.
	\item $\Phi_{\textsc{CS}}$ is the sensing matrix of \textsc{Count-Sketch} with $K = Ck/\epsilon$.
	\item $\Phi_{\textsc{L}}$ is an $R\times N$ matrix, the $r$-th row ($r\in [R]$) of which is 
	\[
		\rho_r = \begin{pmatrix} \eta_{r,1} g_{r,1} & \eta_{r,2} g_{r,2} & \cdots & \eta_{r,n} g_{r,n} \end{pmatrix},
	\]
	where $\eta_{r,i}$ are i.i.d.\@ Bernoulli variables with $\E \eta_{r,i} = 1/(C_0 k)$ and $g_{r,i}$ are i.i.d.\@ $\N(0,1)$ variables.
	\item Each $\Phi_r$ ($r\in[R]$) is a matrix of $4B$ rows defined as follows, where $B = C_Bk/\epsilon$. Let $h_r : [n]\to [B]$ be a random hash function and $\{\sigma_{r,i}\}_{i=1}^n$ be random signs. Define a $B\times n$ hashing matrix $H_r$ as 
	\[
		(H_r)_{j,i} = \begin{cases}	
						(1-\eta_{r,i})\sigma_{r,i}, &i\in h_r^{-1}(j);\\
						0, &\text{otherwise}.
					\end{cases}
	\]
	The $4B$ rows of $\Phi_r$ are defined to be
	\[
		e^{i(\theta_0 \ell_1 + \frac{\pi}{2}\ell_2)} \rho_r  + (H_r)_{b,\cdot}, \qquad  \ell_1,\ell_2 = 0,1,\quad b=1,\dots,B.
	\]
\end{itemize}

\begin{algorithm}[tb]
\begin{algorithmic}[1]
	\algnotext{EndFor}
	\State $S\gets $ the set returned by the algorithm in Theorem~\ref{thm:HH} with $K=k$
	\State Run a \textsc{Count-Sketch} algorithm with $K=Ck/\epsilon$ to obtain an approximation $|\hat x_i|$ to $|x_i|$ for all $i\in S$
	\State $L_r\gets |\sum_{i=1}^n \eta_{r,i} g_{r,i} x_i|$ for all $r\in [R]$\label{alg:ell_infty/ell_2:L}
	\If {$L_1 = 0$}
		\State Run the algorithm for the noiseless case with sparsity $C_2k$
	\Else
		\For {each $r\in [R]$} \label{alg:ell_infty/ell_2:hash_loop_begin}
			\State $S_r' \gets \{i\in S: |\hat x_i|\geq L_r\}$	
			\State $b_{r,i} \gets h_r(i)$ for all $i\in S_r'$
			\State $I_r \gets \{i\in S_r': b_{r,i}\neq b_{r,j}\text{ for all }j\in S_r'\setminus\{i\}\}$
			\For {each $i\in I_r$}
				\State $\tilde\theta_{r,i} \gets $ estimate of phase difference between $x_i$ and $\langle \rho_r,x\rangle$ using Lemma~\ref{lem:all_phase_difference}\label{alg:line:phase_test}
			\EndFor
		\EndFor \label{alg:ell_infty/ell_2:hash_loop_end}
		\State $S'' \gets \bigcap_r S_r'$
		\State Choose an arbitrary $i_0\in S''$
		\For {each $i\in S''\setminus\{i_0\}$}
			\State $\theta_{i_0,i}' = \median_{r: i,i_0\in I_r} (\tilde{\theta}_{r,i}-\tilde{\theta}_{r,i_0})$
		\EndFor
		\For {each $p\in P$}\label{alg:trial_begins}
			\State $\theta_{i_0}'\gets p$
			\State $\theta_i' \gets \theta_{i_0}' + \theta_{i_0,i}'$ for all $i\in S''\setminus\{i_0\}$
			\If {$d(\theta_i',P)\leq \eta/2$ for all $i\in S''$}\label{alg:ell_infty/ell_2:rotational_if}
				\State \Return $\hat x$ supported on $S''$ with $\arg \hat x_i = \theta_i$, where $\theta_i$ is the rounded value of $\theta_i'$ to $P$
			\EndIf
		\EndFor\label{alg:trial_ends}
	\EndIf
\end{algorithmic}
\caption{Algorithm for the $\ell_\infty/\ell_2$ phaseless sparse recovery. Assume that the elements in $P$ are sorted.}\label{alg:ell_inf/ell_2}
\end{algorithm}

We present the recovery algorithm in Algorithm~\ref{alg:ell_inf/ell_2}, where we assume that the set $P$ of valid phases have been sorted. Next we analyse the algorithm in four steps as follows.

\paragraph{Step 1} By Theorem~\ref{thm:HH}, the set $S$ has size $\Oh(k)$ and, with probability $1-1/\poly(n)$, contains all $(1/k)$-heavy hitters. The \textsc{Count-Sketch} (Theorem~\ref{thm:CS}) guarantees that
\begin{equation}\label{eqn:count_sketch}
\left| |x_i| - |\hat{x}_i| \right|^2 \leq  \frac{\epsilon}{C k} \|x_{-Ck/\epsilon} \|_2^2
\end{equation}
for all $i\in S$ with probability at least $1-1/\poly(n)$.

\paragraph{Step 2} Fix a repetition $r\in [R]$. We shall see that $L_r$, calculated in Line~\ref{alg:ell_infty/ell_2:L}, `approximates' the desirable tail $\frac{1}{k}\|x_{-k}\|_2^2$. For notational simplicity, we omit $r$ in the subscript and write $L_{r,i}$ as $L$, $\eta_{r,i}$ as $\eta_i$, etc., since we have fixed the repetition index $r$.

First we upper bound $L$. Decompose $x$ into real and imaginary parts as $x = y + i z$ with $y, z\in \R^n$ and consider $L_1 = \sum_i \eta_i g_i y_i$ and $L_2 = \sum_i \eta_i g_i z_i$. Note that $L^2 = L_1^2 + L_2^2$.

Choosing $C_0 \geq 200$, we have 
\begin{equation}\label{eqn:zero_eta}
\Pr\left\{ \eta_i = 0\text{ for all }i \in H_k(y)\cup H_k(z)\right\}\geq 0.99
\end{equation}
Condition on this event below. Note that $L_1 \sim \N(0,|\sum_i \eta_i y_i|_2^2)$, and
\[
\Pr\left\{ L_1^2 \geq 2.282^2\left|\sum_i \eta_i y_i\right|_2^2 \right\} \leq 1 - 2F(2.282) \leq 0.0225,
\]
where $F(t)$ denotes the cumulative distribution function of the standard normal distribution.

On the other hand, $\E|\sum_i \eta_i y_i|^2 = \sum_i (\E\eta_i) y_i^2 \leq \|y_{-k}\|_2^2/(C_0k)$ thus
\[
\Pr\left\{ \left|\sum_i \eta_i y_i\right|^2 \geq \frac{20}{C_0k}\|x_{-k}\|_2^2 \right\} \leq 0.05,
\]
and hence
\[
\Pr\left\{ L_1^2 \geq  \frac{105}{C_0k}\|y_{-k}\|_2^2 \right\} \leq 0.0725.
\]
Similarly we have
\[
\Pr\left\{ L_2^2 \geq  \frac{105}{C_0k}\|z_{-k}\|_2^2 \right\} \leq 0.0725.
\]
Therefore, taking a union bound of both events above and noting that $\|y_{-k}\|_2^2 + \|z_{-k}\|_2^2 \leq \|x_{-k}\|_2^2$, we have that
\begin{equation}\label{eqn:L_ub}
\Pr\left\{ L^2 \geq  \frac{105}{C_0k}\|x_{-k}\|_2^2 \right\} \leq 0.145.
\end{equation}

We therefore obtained an upper bound of $L$. The next lemma lower bounds $L$.

\begin{lemma}\label{lem:L_lb}
With probability at least $0.8$, it holds that
	$L^2 \geq \frac{1}{C_1 k } \|x_{-C_2 k}\|_2^2$, where $C_1,C_2$ are absolute constants.
\end{lemma}

\begin{proof}
Decompose $x$ into real and imaginary parts as $x = y + i z$ with $y, z\in \R^n$. Consider $L_1 = \sum_i \eta_i g_i y_i$ and $L_2 = \sum_i \eta_i g_i z_i$, which are both real. Note that $L^2 = L_1^2 + L_2^2$.

First consider $L_1$. We sort coordinates $[n] \setminus H_k(y)$ by decreasing order of magnitude. Then, we split the sorted coordinates into continuous blocks of size $C_2' k$ and let $S_j$ denote the $j$-th block. Let $\delta_j$ be the indicator variable of the event that there exists $i \in S_j$ such that $\eta_i = 1$, then $\delta_j$'s are i.i.d.\ Bernoulli variables with $\E\delta_j = 1 - (1-1/(C_0k))^{C_2' k} \geq 1 - \exp(-C_2'/C_0)$, which can be made arbitrary close to $1$ by choosing $C_2'$ big enough. It is a standard result in random walks (see, e.g., \cite[p67]{karlin:first}) that when $\E\delta_j$ is big enough, with probability at least $0.95$, every partial prefix of the 0/1 sequence $(\delta_1, \delta_2, \delta_3, \ldots)$ will have more $1$s than $0$s. Call this event $\mathcal{E}$. In fact, one can directly calculate that $\Pr(\mathcal{E}) = 1 - (1-p)^2/p$ when $p:=\E\delta_j\geq 1/2$, and thus one can take $C_2' = \lceil 1.61 C_0\rceil$ such that $\Pr(\mathcal{E}) \geq 0.95$.

Condition on $\mathcal{E}$. We can then define an injective function $\pi$ from $\{j: \delta_j=0\}$ to $\{j:\delta_j = 1\}$. Specifically, we define $\pi(j) = \ell$, where $\delta_j$ is the $k$-th $0$ in the sequence and $\ell$ is the $k$-th $1$ in the sequence. Clearly $\pi$ is injective, $\pi(j) < j$ and $\delta_{\pi(j)} = 1$. It follows that
\begin{align*}
\sum_{i \in [n] }\eta_i  |y_i|^2  \geq \sum_{ j } \delta_j \| S_{j+1} \|_{\infty}^2 &\geq \sum_{j: \delta_j = 1} \frac{1}{C_2' k}\|S_{j+1}\|_2^2 \\
&\geq  \frac{1}{2} \sum_{j:\delta_j=1} \frac{1}{C_2' k}\|S_{j+1}\|_2^2 + \frac{1}{2} \sum_{\substack{j:\delta_j=1\\ \pi^{-1}(j)\text{ exists}}}   \frac{1}{C_2' k}\|S_{\pi^{-1}(j)}\|_2^2 \\
&\geq  \frac{1}{2} \sum_{j:\delta_j=1} \frac{1}{C_2' k}\|S_{j+1}\|_2^2 + \frac{1}{2} \sum_{j:\delta_j=0}  \frac{1}{C_2' k}\|S_{j}\|_2^2 \\
&\geq \frac{1}{2 C_2' k} \|y_{-C_2' k}\|_2^2.
\end{align*}
This implies that $L_1 = \sum_i \eta_i g_i y_i\sim \N(0,\sum_i \eta_i|y_i|^2)$ with probability at least $0.95$ will stochastically dominate a gaussian variable $\N(0,\frac{1}{2C_2k} \|y_{-C_2k}\|_2^2)$. Combining with the fact that $\Pr_{g\sim \N(0,1)}\{|g|\leq \frac{1}{16}\}\leq 0.05$, we see that 
\[
\Pr\left\{ L_1^2 \geq  \frac{1}{16^2\cdot 2C_2 k} \|y_{-C_2' k}\|_2^2\right\} \geq 0.9.
\]
Similarly for the imaginary part $z$ and $L_2$,
\[
\Pr\left\{ L_2^2 \geq  \frac{1}{16^2 \cdot 2C_2 k} \|z_{-C_2' k}\|_2^2\right\} \geq 0.9.
\]
Condition on that both events above happen. For notational convenience, let $T_1 = H_{C_2' k}(y)$ and $T_2 = H_{C_2' k}(z)$, then
\begin{align*}
L^2 = L_1^2 + L_2^2 &\geq \frac{1}{16^2 \cdot 2 C_2' k} (\|y_{T_1^c}\|_2^2 + \|z_{T_2^c}\|_2^2)\\
&\geq \frac{1}{16^2 \cdot 2 C_2' k} \|y_{(T_1\cup T_2)^c}\|_2^2 + \|z_{(T_1\cup T_2)^c}\|_2^2)\\
&= \frac{1}{16^2\cdot 2 C_2' k} \|x_{(T_1\cup T_2)^c}\|_2^2\\
&\geq \frac{1}{16^2\cdot 2 C_2' k} \|x_{-2C_2' k}\|_2^2.
\end{align*}
Therefore, we can take $C_2 = 2C_2'$ above and $C_1 = 16^2 C_2$.
\end{proof}

Combining \eqref{eqn:zero_eta}, \eqref{eqn:L_ub} and Lemma~\ref{lem:L_lb} and taking $C_0= 210$, we conclude that with probability at least $1 - 0.365$,
\begin{equation}\label{eqn:L_both}
\frac{1}{C_1k}\|x_{-C_2k}\|_2^2 \leq L^2\leq \frac{1}{2k}\|x_{-k}\|_2^2.
\end{equation}

\paragraph{Step 3} We keep the repetition $r$. We now show that the trimmed set $S'$ is good in the sense that its elements are not too small and it contains all $(1/k)$-heavy hitters. This is formalized in the following lemma.
\begin{lemma}\label{lem:S'}
Suppose that event \eqref{eqn:count_sketch} happens. With probability at least $0.635$, it holds that
\begin{enumerate}[(i)]
	\item $|x_i|^2 \geq \frac{1}{2C_1k} \|x_{-C_2k}\|_2^2 $ for all $i\in S'$; and
	\item $S'$ contains all coordinates $i$ such that $|x_i|^2 \geq \frac{1}{k} \|x_{-k}\|_2^2$.
\end{enumerate}
\end{lemma}

\begin{proof}
The events \eqref{eqn:zero_eta} and \eqref{eqn:L_both} happen simultaneously with probability at least $1-0.365$. Condition on both events. Let $C = \frac{\sqrt2}{\sqrt2-1}C_1$. By our choice of constants, $C\geq C_1\geq C_2$, then
\begin{enumerate}[(i)]
\item for $i \in S'$, it holds that $|x_i| \geq L - \frac{1}{\sqrt{Ck}} \|x_{-C_2k}\|_2 \geq \frac{1}{\sqrt{2C_1k}} \|x_{-C_2k} \|_2$; 
\item if $|x_i|^2 \geq \frac{1}{k} \|x_{-k}\|_2^2$, then $|\hat{x}_i| \geq \frac{1}{\sqrt{k}} \|x_{-k}\|_2 - \frac{1}{\sqrt{Ck}} \|x_{-k}\|_2 \geq L$.\qedhere
\end{enumerate}
\end{proof}

\paragraph{Step 4} We now show that each $\tilde{\theta}_{r,i}$ is good for $i\in I_r$. Let $\theta_{r,i}$ be the (oriented) phase difference between $x_i$ and $\sum_{j=1}^n \eta_{r,j}g_{r,j}x_j$.
\begin{lemma}\label{lem:good theta tilde}
Suppose that event in Lemma~\ref{lem:S'} happens for $S_r'$. Then it holds for each $i\in I_r$ that, with probability at least $0.95$, $|\tilde{\theta}_{r,i} - \theta_{r,i}| \leq c\epsilon$ for some absolute constant $c$.
\end{lemma}
\begin{proof}
We have from our construction the measurements
\[
\left|e^{i(2c\epsilon \ell_1 + \frac{\pi}{2}\ell_2)}\sum_{i=1}^n \eta_{r,i} g_{r,i} x_i + \sum_{i \in h_{r}^{-1}(j) } (1-\eta_{r,i}) \sigma_{r,i} x_i\right|, \quad \ell_1,\ell_2 = 0,1,\ j\in[B],\ r\in[R].
\]

We note that Line~\ref{alg:line:phase_test} in the algorithm is valid because we have access to 
\[
\left|\hat{x}_i\right|,\ \left|\sum_{i=1}^n \eta_{r,i} g_{r,i} x_{r,i}\right|,\ \left|\hat{x}_i + e^{i(2c\epsilon \ell_1 + \frac{\pi}{2}\ell_2)}\sum_{i=1}^n \eta_{r,i} \sigma_{r,i} x_i + \sum_{ i' \in h^{-1}_r(h_r(i))\setminus S'} (1-\eta_{r,i'}) g_{r,i'} x_{i'}\right|.
\]

Let
\[
Z = \sum_{ i' \in h^{-1}_r(h_r(i))\setminus S'} (1-\eta_{r,i'}) \sigma_{r,i'} x_{i'}.
\]
Observe that 
\[
\Pr\left\{\left|h^{-1}_{r}(h_r(i)) \cap H_{_{C_2 k}}(x) \right| = 1\right\}\geq 1 - \frac{C_2}{C_B}
\]
and that 
\[
\E |Z|^2 = \frac{\epsilon}{C_Bk } \| x_{_{-C_2k}} \|_2^2.
\]
By Markov's inequality, we have that $|Z|^2 \leq \frac{40\epsilon}{C_B k} \|x_{-k}\|_2^2$ with probability at least $0.025$. Choose $C_B$ such that $\frac{40}{C_B} \leq \frac{1}{2C_1}$ and $\frac{C_2}{C_B} < \frac{1}{40}$, the assumptions on noise magnitude in Lemma~\ref{lem:all_phase_difference} will be satisfied for $x_i$ with probability at least $0.95$. Then  Lemma~\ref{lem:all_phase_difference} yields an estimate $\tilde\theta_{r,i}$ which satisfies $|\tilde\theta_{r,i}-\theta_i|\leq c\epsilon$.
\end{proof}

\paragraph{Step 5} We are now ready to prove Theorem~\ref{thm:ell_infty_ell_2}.

\begin{proof}[Proof of Theorem~\ref{thm:ell_infty_ell_2}]
Let $i\in S''$. Observe that for each $r$, it is isolated from every other $i' \in S_r'$ with probability $\frac{1}{C_B}$. A simple Chernoff bound shows that with probability at least $1-1/\poly(n)$, it is contained in at least $0.9$ fraction of $\{I_r\}_{r\in [R]}$, provided that $C_B$ is big enough.

Suppose that event~\eqref{eqn:count_sketch} happens. It follows from Lemmata~\ref{lem:S'} and~\ref{lem:good theta tilde} that in each repetition $r\in [R]$, for each $i\in I_r\setminus\{i_0\}$ it holds with probability at least $1 - 0.465$  that 
\[
|(\tilde\theta_{r,i}-\tilde\theta_{r,i_0}) - (\theta_{r,i}-\theta_{r,i_0})|\leq 2c\epsilon.
\]
Next we take the median over valid repetitions. To do this, we first check if there are at least half of $\tilde\theta_{r,i}-\tilde\theta_{r,{i_0}}$, when interpreted as real numbers in $[-\pi,\pi)$, satisfy that $\tilde\theta_{r,i}-\tilde\theta_{r,{i_0}} \in [-\frac{3}{25}\pi,\frac{3}{25}\pi]$. If this happens, take $\theta'_{i_0,i}$ to be the median of those numbers in $[-\pi,\pi)$; otherwise, take $\theta'_{i_0,i}$ to be the median with the interpretation of all $\tilde\theta_{r,i}-\tilde\theta_{r,{i_0}}$ as real numbers in $[0,2\pi)$. Since both $i$ and ${i_0}$ appear in at least $0.8$ fraction of $\{I_r\}_{r\in [R]}$, we can, by choosing the constant $c_R$ large enough, guarantee that
\begin{equation}\label{eqn:tilde_theta_i}
	\left|\theta'_{i_0,i} - (\theta_{r,i}-\theta_{r,i_0})\right| \leq 2c\epsilon 
\end{equation}
with probability at least $1-1/\poly(n)$. This allows for a union bound over all $i\in S''\setminus \{i_0\}$. Therefore \eqref{eqn:tilde_theta_i} holds for all $i\in S''\setminus\{i_0\}$ simultaneously with probability $\geq 1-1/\poly(n)$.

Next, assume that it happens that \eqref{eqn:tilde_theta_i} holds for all $i\in S''\setminus\{i_0\}$. Consider the for-loop from Line~\ref{alg:trial_begins} to \ref{alg:trial_ends}. It is clear that when $\theta_{i_0}'$ is exactly the phase of $x_{i_0}$, it will hold that $\theta_i'$ is an accurate estimate to the phase of $x_i$ up to an additive error of $2c\epsilon < \eta/2$. The if-clause in Line \ref{alg:ell_infty/ell_2:rotational_if} will be true and the algorithm will terminate with an $\hat x$. Since the phases of the entries are at least $\eta$ apart, there will be no ambiguity in rounding and the phases in $\hat x$ are all correct, hence the error $\|x-\hat x\|_2$ only depends on the magnitude errors, which is exactly \eqref{eqn:count_sketch}, obtained from applying \textsc{Count-Sketch}. When $\theta_{i_0}'$ is not $x_{i_0}$, by the rotational $(k,\eta)$-distinctness, $\{\theta_i'\}$ will coincide with $P$ exactly or the if-clause will not be true. This shows the correctness. 

Removing the conditioning of~\eqref{eqn:count_sketch} increases the failure probability $1/\poly(n)$ and the overall failure probability is $1/\poly(n)$.

\paragraph*{Number of Measurements} The submatrix $\Phi_{\textsc{HH}}$ has $\Oh(k\log n)$ rows, the submatrix $\Phi_{\textsc{CS}}$ has $\Oh((k/\epsilon)\log n)$ rows, the submatrix $\Phi_{\textsc{L}}$ has $\Oh(\log n)$ rows, 
each $\Phi_r$ for $r\in [R]$ has $\Oh(k/\epsilon)$ rows. Hence the total number of rows is dominated by that of $\Phi_{\textsc{CS}}$ and the $R$ independent copies of $\Phi_r$'s, that is, $\Oh((k/\epsilon)\log n + R(k/\epsilon)) = \Oh((k/\epsilon)\log n) = \Oh((k/\eta)\log n)$.

\paragraph*{Runtime} Line 1 runs in time $\Oh(k\poly(\log n))$, Line 2 in time $\Oh(|S|\log n) = \Oh(k\log n)$, Line 3 in time $O(\log n)$. The runtime before the if-branch of Line 4 is thus $\Oh(k\poly(\log n))$.

The noiseless case in Line 5 runs in time $\Oh(k\log k)$. 

For the for-loop from Line \ref{alg:ell_infty/ell_2:hash_loop_begin} to \ref{alg:ell_infty/ell_2:hash_loop_end}, Line 8 runs in time $\Oh(k)$, Line 9 in $\Oh(k)$, Line 10 in $\Oh(k\log k)$, Lines 11--12 in time $O(k)$. Hence this for-loop takes time $\Oh(R k\log k) = \Oh(k\poly(\log n))$. Line 13 runs in time $O(Rk) = O(k\log n)$. Calculating each median takes time $O(R)$ and thus Lines 15--16 takes time $O(Rk) = O(k\log n)$. The for-loop from Line \ref{alg:trial_begins} to \ref{alg:trial_ends} takes time $\Oh(k/\eta)$, since $|P|=\Oh(1/\eta)$, the if-clause in Line \ref{alg:ell_infty/ell_2:rotational_if} can be verified in time $\Oh(k)$ if the elements in $P$ are sorted in advance. Hence the total runtime of Lines 7--22 is $\Oh(k/\eta + k\poly(\log n))$.

Therefore, the overall runtime is $\Oh(k/\eta + k\poly(\log n))$.
\end{proof}

\section{Future Work}

In this paper we obtained sublinear-time algorithms for different versions of the compressive phase retrieval problem, using purely combinatorial techniques. Our results include the first sublinear-time algorithm for the $\ell_{\infty}/\ell_2$ problem, a popular problem in the literature. We suggest future directions and open questions that may be of interest to the computer science community.

\begin{itemize}
\item A natural question is to generalize our results to a more general set of valid phases. The difficulty lies in the fact that we can only estimate the phase up to an additive error of $O(\eta)$ using the Law of Cosines, and this would incur a large error if the corresponding heavy hitter has a huge magnitude.

\item Is it possible to obtain a uniform guarantee for noiseless signals that runs in $\tilde{\Oh}(k)$ time and uses $\Oh(k)$ measurements? Even something that is substantially less than $\Oh(k^2)$ time would be interesting.

\item Another question is whether one can achieve almost optimal bounds using structured measurements, one example being local correlation measurements as in \cite{IVW16}. This would pave the way of tackling the more constrained problem of phase retrieval, where the sensing matrix is the Discrete Fourier Matrix, a problem of great importance in engineering. There are numerous directions and open problems in this direction; interested readers can refer to~\cite{jaganathan2015phase,SECCMS15}.  
\end{itemize}


\section*{Acknoweledgement}

We would like to thank the anonymous reviewers for their comments and suggestions that greatly helped improve and clarify the paper.

\bibliographystyle{plain}
\bibliography{bibio}

\end{document}